\RequirePackage{amsthm}
\documentclass[pdflatex]{sn-jnl}
\usepackage[numbers]{natbib}

\usepackage{amsmath}
\usepackage{cleveref}
\crefname{algocf}{Alg.}{Algs.}
\Crefname{algocf}{Algorithm}{Algorithms}
\usepackage[linesnumbered,ruled,vlined]{algorithm2e}
\usepackage{mathtools}
\newcommand{\alert}[1]{\textcolor{red}{#1}}
\newcommand{\true}{\texttt{true}}
\newcommand{\false}{\texttt{false}}
\newcommand{\OPT}{\mathrm{OPT}}

\DeclareMathOperator*{\argmax}{arg\,max}

\newcommand{\cS}{\mathcal{S}}
\renewcommand{\mid}{:}
\usepackage{comment}
\usepackage{lmodern}
\usepackage[T1]{fontenc}
\usepackage{moresize}

\usepackage{tabularx}
\newcolumntype{C}{>{\centering\arraybackslash}X}
\NewExpandableDocumentCommand\mcc{O{1}m}{\multicolumn{#1}{c}{#2}}
\newcolumntype{S}{>{\color{blue}\scriptsize}c}
\newcolumntype{T}{>{\color{blue}\ssmall}c}
\newcommand{\noline}[1]{\multicolumn{1}{c}{\textcolor{blue}{\scriptsize #1}}}
\newcolumntype{Q}{@{}>{\centering\arraybackslash}p{2em}@{}}
\usepackage{hhline}

\usepackage{graphicx}%
\usepackage{multirow}%
\usepackage{amsmath,amssymb,amsfonts}%
\usepackage{amsthm}%
\usepackage[title]{appendix}%
\usepackage{xcolor}%
\usepackage{textcomp}%
\usepackage{manyfoot}%

\theoremstyle{thmstyleone}%
\newtheorem{theorem}{Theorem}
\theoremstyle{thmstyletwo}%
\newtheorem{example}{Example}%

\theoremstyle{thmstylethree}%
\newtheorem{lemma}{Lemma}
\newtheorem{corollary}{Corollary}
\begin{document}
\title[Article Title]{Resource Allocation under the Latin Square Constraint}

\author[1]{\fnm{Yasushi} \sur{Kawase}}

\author[2]{\fnm{Bodhayan} \sur{Roy}}

\author[2]{\fnm{Mohammad Azharuddin} \sur{Sanpui}}

\affil[1]{
\orgname{The University of Tokyo}, 
\orgaddress{
\state{Tokyo}, \country{Japan}}}

\affil[2]{
\orgname{Indian Institute of Technology Kharagpur}, %
\orgaddress{
\state{West Bengal}, \country{India}}}

\abstract{A Latin square is an $n \times n$ matrix filled with $n$ distinct symbols, each of which appears exactly once in each row and exactly once in each column.
We introduce a problem of allocating $n$ indivisible items among $n$ agents over $n$ rounds while satisfying the Latin square constraint. This constraint ensures that each agent receives no more than one item per round and receives each item at most once. Each agent has an additive valuation on the item--round pairs. 
Real-world applications like scheduling, resource management, and experimental design require the Latin square constraint to satisfy fairness or balancedness in allocation.
Our goal is to find a partial or complete allocation that maximizes the sum of the agents' valuations (utilitarian social welfare) or the minimum of the agents' valuations (egalitarian social welfare). 
For the problem of maximizing utilitarian social welfare, we prove NP-hardness even when the valuations are binary additive. 
We then provide $(1-1/e)$ and $(1-1/e)/4$-approximation algorithms for partial and complete settings, respectively. Additionally, we present fixed-parameter tractable (FPT) algorithms with respect to the order of Latin square and the optimum value for both partial and complete settings. 
For the problem of maximizing egalitarian social welfare, we establish that deciding whether the optimum value is at most $1$ or at least $2$ is NP-hard for both the partial and complete settings, even when the valuations are binary. 
Furthermore, we demonstrate that checking the existence of a complete allocation that satisfies each of envy-free, proportional, equitable, envy-free up to any good, proportional up to any good, or equitable up to any good is NP-hard, even when the valuations are identical.}




\maketitle

\section{Introduction}\label{sec1}
The fair division of indivisible resources constitutes a significant and complex challenge at the intersection of economics, mathematics, and computer science, bearing substantial implications for both theoretical and practical applications \cite{Bouveret2016FairAO,lipton2004approximately,alkan1991fair,brams1996fair,klamler2010fair,Cake,Fairdivisionandcollectivewelfare}. In contrast to divisible resources---such as land, capital, or commodities---that can be divided according to the preferences of agents \cite{aziz2016discrete,aziz2020bounded,goldberg2020contiguous,segal2017fair}, indivisible resources cannot be divided without substantially reducing their utility or worth \cite{chevaleyre2017distributed,segal2019democratic,aziz2022fair,aziz2023possible}. Examples encompass residences, automobiles, or artworks, wherever fractional ownership is either impractical or undesirable.

Allocating indivisible resources fairly and efficiently is crucial in many practical scenarios, such as scheduling sightseeing for multiple groups visiting various locations or assigning shifts to medical professionals.  

In this paper, we introduce the problem of allocating $n$ indivisible items among $n$ agents over $n$ rounds, ensuring that each agent receives each item once.
This problem can be regarded as an allocation problem under a \emph{Latin square} constraint.

The notion of indivisible item allocation constrained by a Latin square provides a systematic and effective method for distributing tasks or resources in many real-world contexts. A Latin square is a mathematical configuration in which each element occurs exactly once in each row and each column of a grid, thus preventing any repetition within the same context (see Figure~\ref{fig:exLS}). 
Various domains, such as job scheduling, school timetabling, resource allocation optimization in computing systems, and event seating arrangement structuring, utilize this constraint. See \Cref{subsec:application} for more details on examples of applications.

A \emph{complete Latin square} of order $n$ is an $n\times n$ array filled with $n$ different symbols, each occurring exactly once in each row and exactly once in each column. 
A \emph{partial Latin square} is the case where some cells may be empty.
Figure~\ref{fig:exLS} illustrates two examples of complete Latin squares and two examples of partial Latin squares.

\begin{figure}[htbp]
\renewcommand{\arraystretch}{1.3}
\begin{tabular}{|Q|Q|Q|Q|}
\hline
A & B & C & D \\\hline
B & C & D & A \\\hline
C & D & A & B \\\hline
D & A & B & C \\\hline
\end{tabular}
\quad
\begin{tabular}{|Q|Q|Q|Q|}
\hline
D & A & C & B \\\hline
C & B & D & A \\\hline
B & C & A & D \\\hline
A & D & B & C \\\hline
\end{tabular}
\quad
\begin{tabular}{|Q|Q|Q|Q|}
\hline
A & B &   &   \\\hline
B & A &   &   \\\hline
  &   & D & C \\\hline
  &   & C & D \\
\hline
\end{tabular}
\quad
\begin{tabular}{|Q|Q|Q|Q|}
\hline
A & B & C &   \\\hline
  &   &   & D \\\hline
  &   &   &   \\\hline
  &   &   &   \\
\hline
\end{tabular}
\caption{Examples of complete and partial Latin squares}\label{fig:exLS}
\end{figure}
In our setting, rows correspond to items, columns correspond to rounds, and symbols correspond to agents.
The structure of the Latin square ensures that no item is allocated to multiple agents in each round, 
each agent receives at most one item per round, and no agent receives the same item more than once.

Suppose that each agent $i$ has a valuation $v_{ijk}$ for each pair of item $j$ and round $k$. The utility of agent $i$ is defined as the sum of the valuations that they receive.
We investigate the computational complexities of finding a partial or complete allocation under the Latin square constraint that maximizes social welfare. We call this problem the \emph{Latin square allocation (LSA)} problem.
As the measure of social welfare, we employ two settings: utilitarian social welfare and egalitarian social welfare.
Utilitarian social welfare is defined as the sum of the utilities of the agents, while egalitarian social welfare is defined as the minimum of the utilities of the agents.
\subsection{Related Work}
Latin squares have been the subject of various studies in algebra and combinatorics and have applications in fields such as mathematical puzzles, coding theory, and experimental design~\cite{keedwell2015latin,zhou1994disk,bao2004mals,laywine1998discrete}. 
The study of Latin squares has a rich history and has also been studied from computational aspects. 
One of the fundamental problems in this area is the completion of partial Latin squares.
Colbourn~\cite{colbourn1984complexity} proved NP-hardness of this problem.
Kumar et al.~\cite{kumar1999approximating} introduced a maximization version of completing a partial Latin square. 

Finding an allocation that maximizes utilitarian social welfare or egalitarian social welfare under a constraint has been extensively studied in the context of fair and efficient allocation~\cite{Golovin2005,chakrabarty2009allocating,cohler2011optimal,aumann2012computing,kawase2024contiguous,bei2012optimal,kawase2020max,kawase2024minimizing}.
Additionally, the problem of maximizing utilitarian social welfare is well-explored in the context of combinatorial auctions~\cite{Bertelsen2005,CG2010,Vondrak2008,KLMM2008,AGT}.
Furthermore, maximizing egalitarian social welfare is also called \emph{max-min fair}, and a special case is extensively studied under the name of the \emph{Santa Claus} problem~\cite{BansalS2006,BezakovaD2005,AS2010,Haeupler2011}.


Resource allocation under the Latin square constraint can be seen as a \emph{multi-layered cake cutting} problem. 
The multi-layered cake cutting problem, introduced by Hosseini et al.~cite{Hosseini2020FairDO}, involves allocating layers of a cake without overlap. Similarly, in our resource allocation problem under the Latin square constraint, each round acts as a layer, ensuring no agent receives more than one item per layer. This parallels the non-overlapping constraint in multi-layered cake cutting. Our study of resource allocation under the Latin square constraint was inspired by this multi-layered cake cutting problem with non-overlapping constraint. While these problems are distinct, they share similarities in layers and non-overlapping constraints. However, no previous results from cake cutting directly impact our findings on resource allocation under the Latin square constraint.
Igarashi and Fr{\'e}d{\'e}ric~\cite{igarashi2021envy} further explored the problem and proved the existence of such an allocation in a more general setting. Several other papers have contributed to the understanding of the multi-layered cake cutting problem~\cite{Cloutier2009TwoplayerEM,Nyman2017FairDW,Lebert2013EnvyfreeTM}.

Moreover, resource allocation under the Latin square constraint can also be seen as a \emph{repeated allocation over time}, in which a set of items is allocated to the same set of agents repeatedly over multiple rounds.
Several studies discuss the existence of fair and efficient repeated allocations, as well as the computational complexity involved in finding a desired repeated allocation~\cite{micheel2024fairness,caragiannis2024repeatedly,trabelsi2023resource,elkind2022fairness,igarashi2024repeated}.
However, unlike our setting, each agent can receive the same item more than once.

\subsection{Our Results}
We study the computational complexity of finding partial and complete allocations that satisfy a given efficiency or fairness criterion while adhering to the Latin square constraint.

In \Cref{sec:approx-Umax}, we explore approximation algorithms for the LSA problem of maximizing utilitarian social welfare.
We first provide a $(1-1/e)$-approximation algorithm for the partial LSA problem. Next, we present a $(1-1/e)/4$-approximation algorithm for the complete LSA problem by showing that an $\alpha$-approximate solution for the partial LSA problem can be converted into ${\alpha}/{4}$-approximate solution of the complete LSA problem.

In \Cref{sec:FPT-Umax}, we construct two fixed-parameter tractable (FPT) algorithms for addressing the problem of maximizing utilitarian social welfare with respect to the order of Latin square and the optimum value, respectively, for both partial and complete settings. 
We can also construct an FPT algorithm with respect to the order of Latin square for maximizing egalitarian social welfare.

In \Cref{sec:NP-hard}, we show that the LSA problems for maximizing utilitarian social welfare of the complete and partial settings are both NP-hard, even when the valuations are binary. Additionally, we establish that deciding whether the maximum egalitarian social welfare is at most $1$ or at least $2$ is NP-hard for both the partial and complete settings, even when the valuations are binary additive. 
Furthermore, we construct an approximation-preserving reduction from the max-min fair allocation problem to the LSA problems for maximizing egalitarian social welfare.
The hardness result and reduction suggest that it is unlikely to exist a polynomial-time algorithm with a good approximation ratio or an FPT algorithm with respect to the optimal value for the LSA problems for maximizing egalitarian social welfare.
Furthermore, we demonstrate that checking the existence of a complete allocation that satisfies each of the following conditions is NP-hard, even when the valuations are identical: envy-free (EF), equitable (EQ), proportional (PROP), envy-free up to any good (EFX), equitable up to any good (EQX), and proportional up to any good (PROPX).
For the definitions of these properties, see the last paragraph of \Cref{sec:preliminaries}.

\section{Preliminaries}\label{sec:preliminaries}
\subsection{Model}
For a positive integer $n$, we denote the set $\{1,2,\dots,n\}$ by $[n]$.
In this paper, we address the problem of assigning $n$ items to $n$ agents over $n$ rounds, which we refer to LSA problem. 
Let $N=[n]$ be the set of agents, $M=[n]$ be the set of items, and $R=[n]$ be the set of rounds.
An allocation $A$ is a subset of triplets $N\times M\times R$ where $(i,j,k)\in A$ means that agent $i\in N$ receives item $j\in M$ in round $k\in R$.
An allocation $A$ is \emph{feasible} if 
\begin{itemize}
\item (i) each agent receives at most one item per round (i.e., $|\{(i,j,k')\in A\mid k'\in R\}|\le 1$ for each $i\in N$ and $j\in M$),

\item (ii) no item is allocated to multiple agents in each round (i.e., $|\{(i,j',k)\in A\mid j'\in M\}|\le 1$ for each $i\in N$ and $k\in R$), and

\item (iii) no agent receives the same item more than once (i.e., $|\{(i',j,k)\in A\mid i'\in N\}|\le 1$ for each $j\in M$ and $k\in R$).
\end{itemize}
For a feasible allocation $A$, we write $A(j,k)$ to denote the agent who receives item $j\in M$ in round $k\in R$ if such an agent exists and $A(j,k)=\bot$ if no such agent exists.
We will call a pair of an item and a round a cell.
Additionally, we will denote by $A_i$ the set $\{(j,k)\in M\times R\mid A(j,k)=i\}$.
A feasible allocation where each agent receives each item exactly once (i.e., $|A|=n^2$) is called a \emph{complete} allocation. 
Note that a (complete) allocation corresponds to a (complete) Latin square. 
We will refer to a feasible allocation that is not necessarily complete as \emph{partial}.

Each agent $i\in N$ gets a non-negative integer value of $v_{ijk}\in\mathbb{Z}_{+}$ when receiving item $j\in M$ in round $k\in R$. We say that the valuations are \emph{binary} if $v_{ijk}\in\{0,1\}$ for all $i\in N$, $j\in M$, and $k\in R$. Additionally, we say that the valuations are \emph{identical} if $v_{ijk}=v_{i'jk}$ for all $i,i'\in N$, $j\in M$, and $k\in R$.
For $S\subseteq M\times R$, let $v_i(S)$ denote the value $\sum_{(j,k)\in S} v_{ijk}$.
The \emph{utilitarian social welfare} and the \emph{egalitarian social welfare} of a feasible allocation $A$ is defined as $\sum_{i\in N} v_i(A_i)$ and $\min_{i\in N}v_i(A_i)$, respectively.
We call an allocation $A$ \emph{Umax} and \emph{Emax} if it maximizes utilitarian social welfare and egalitarian social welfare, respectively.

The complete Umax LSA problem is a problem of finding a complete allocation that maximizes utilitarian social welfare. This problem can be represented as the following integer linear programming (ILP):
\begin{align}
\begin{array}{rll}
\max& \sum_{i\in N} \sum_{j\in M} \sum_{k\in R} v_{ijk}x_{ijk}&\\
\text{s.t.}& \sum_{i\in N}x_{ijk}=1&\hspace{-3mm}(\forall j\in M,\,k\in R),\\
& \sum_{j\in M}x_{ijk}=1&\hspace{-3mm}(\forall k\in R,\,i\in N),\\
& \sum_{k\in R}x_{ijk}=1&\hspace{-3mm}(\forall i\in N,\,j\in M),\\
& x_{ijk}\in\{0,1\}     &\hspace{-3mm}(\forall i\in N,\,j\in M,\,k\in R),
\end{array}\label{eq:completeILP}
\end{align}
where $A=\{(i,j,k)\in N\times M\times R\mid x_{ijk}=1\}$ corresponds to a complete allocation.
The partial Umax LSA problem is a problem of finding a partial allocation that maximizes utilitarian social welfare, which can be represented as the following ILP:
\begin{align}
\begin{array}{rll}
\max& \sum_{i\in N} \sum_{j\in M} \sum_{k\in R} v_{ijk}x_{ijk}&\\
\text{s.t.}& \sum_{i\in N}x_{ijk}\le 1&\hspace{-3mm}(\forall j\in M,\,k\in R),\\
& \sum_{j\in M}x_{ijk}\le 1&\hspace{-3mm}(\forall k\in R,\,i\in N),\\
& \sum_{k\in R}x_{ijk}\le 1&\hspace{-3mm}(\forall i\in N,\,j\in M),\\
& x_{ijk}\in\{0,1\}     &\hspace{-3mm}(\forall i\in N,\,j\in M,\,k\in R).
\end{array}\label{eq:partialILP}
\end{align}
We define the partial and complete Emax LSA problems in a similar manner.

It is important to note that the Umax and Emax values of a partial LSA problem are at least as large as those of the corresponding complete LSA problem, respectively. This is because any complete allocation is also a partial allocation. Furthermore, as demonstrated below, the Umax and Emax values of a partial LSA problem can be strictly greater than those of the corresponding complete LSA problem, even for the binary case. 
\begin{example}\label{ex:noexist}
Suppose that $n=2$, $v_{1,1,1}=v_{2,2,2}=1$, and $v_{1,1,2}=v_{1,2,1}=v_{1,2,2}=v_{2,1,1}=v_{2,1,2}=v_{2,2,1}=0$.
Then, the Umax value is $2$, and the Emax value is $1$ for the partial LSA problem, achieved by the partial allocation $\{(1,1,1),(2,2,2)\}$.
On the other hand, there are only two complete allocations: $\{(1,1,1),(1,2,2),(2,1,2),(2,2,1)\}$ and $\{(1,1,2),(1,2,1),(2,1,1),(2,2,2)\}$. Both of these complete allocations have a utilitarian social welfare of $1$ and an egalitarian social welfare of $0$.
Thus, a partial allocation can yield a higher optimum value compared to every complete allocation.
\end{example}


An allocation $A$ is called \emph{envy-free (EF)} if every agent evaluates that their allocated bundle as at least as good as any other agent's bundle, i.e., $v_i(A_i)\ge v_i(A_{i'})$ for any $i,i'\in N$.
Similarly, \emph{envy-free up to one good (EF1)} and \emph{envy-free up to any good (EFX)} are defined as follows:
\begin{itemize}
    \item EF1: $v_i(A_i)\ge v_i(A_{i'}\setminus \{(j,k)\})$ for some $(j,k)\in A_{i'}$ or $A_{i'}=\emptyset$ $(\forall i,i'\in N)$,
    \item EFX: $v_i(A_i)\ge v_i(A_{i'}\setminus\{(j,k)\})$ for any $(j,k)\in A_{i'}$ $(\forall i,i'\in N)$.\footnotemark{}
\end{itemize}
An allocation $A$ is said to be \emph{proportional (PROP)}, \emph{proportional up to one good (PROP1)}, and \emph{proportional up to any good (PROPX)} as follows:
\begin{itemize}
\item  PROP: $v_i(A_i)\ge v_i(M\times R)/n~(\forall i\in N)$, 
\item PROP1: $v_i(A_i)\ge v_i(M\times R)/n-\max_{(j,k)\in (M\times R)\setminus A_i}v_{ijk}~(\forall i\in N)$, 
\item PROPX: $v_i(A_i)\ge v_i(M\times R)/n-\min_{(j,k)\in (M\times R)\setminus A_i}v_{ijk}~(\forall i\in N)$. 
\end{itemize}
An allocation $A$ is said to be \emph{equitable (EQ)}, \emph{equitable up to one good (EQ1)}, and \emph{equitable up to any good (EQX)} as follows:
\begin{itemize}
\item EQ: $v_i(A_i)=v_{i'}(A_{i'})~(\forall i,i'\in N)$, 
\item EQ1: $v_i(A_i)\ge v_{i'}(A_{i'})-\max_{(j,k)\in A_{i'}}v_{i'jk}~(\forall i,i'\in N)$, 
\item EQX: $v_i(A_i)\ge v_{i'}(A_{i'})-\min_{(j,k)\in A_{i'}}v_{i'jk}~(\forall i,i'\in N)$.
\end{itemize}

\footnotetext{This definition of EFX follows a stronger variant introduced by Plaut and Roughgarden~\cite{plaut2020almost}. A weaker variant of EFX can be defined according to the original definition of Caragiannis et al.~\cite{caragiannis2019unreasonable} as follows: $v_i(A_i)\ge v_i(A_{i'}\setminus\{(j,k)\})$ for any $(j,k)\in A_{i'}$ \emph{with $v_{i'jk}>0$}. The non-existence of a complete allocation satisfying this weaker variant of EFX can also be derived from \Cref{ex:noexist}.}

\begin{example}\label{ex:noexist2}
Suppose that $n=2$, $v_{i,1,1}=v_{i,2,2}=1$, and $v_{i,1,2}=v_{i,2,1}=0$ for each $i\in\{1,2\}$.
Then, no complete allocation is EF, EF1, EFX, PROP, PROP1, PROPX, EQ, EQ1, or EQX.
\end{example}

\subsection{Applications}\label{subsec:application}
The Latin square constraint appears frequently in scheduling and resource allocation problems.
We will now discuss some examples.



\subsubsection*{Sightseeing Scheduling} 
In scheduling sightseeing for multiple groups, it is necessary to prevent multiple groups from visiting the same location at the same time to reduce congestion.
Suppose that there are $n$ groups $N$ who want to visit $n$ locations $M$ over $n$ time slots $R$. 
Then, the Latin square constraint guarantees that:
\begin{itemize}
\item Each group visits at most one location per round.
\item No location is visited by multiple groups in each round.
\item No agent visits the same location more than once.
\end{itemize}
Each group $i\in N$ has a valuation $v_{ijk}$ for each pair of location $j\in M$ and time slot $k\in R$. 
In this situation, solving the complete/partial Umax LSA problem leads to an allocation that maximizes utilitarian social welfare.

\subsubsection*{School Timetabling}
Educational institutions frequently require the allocation of classes or teachers to groups of students in a manner that ensures no clash of student groups, teachers, or time slots \cite{pillay2010overview,pillay2014survey,hilton1980reconstruction}. 
Suppose that there are $n$ student groups $N$, $n$ classes (or teachers) $M$, and $n$ time slots $R$. 
The Latin square constraint guarantees that:
\begin{itemize}
    \item Each student group attends at most one class at a time.
    \item No class is assigned to multiple student groups in each time slot. 
    \item No student group attends the same class more than once.
\end{itemize}
Each student group $i\in N$ has a valuation $v_{ijk}$ for each pair of class $j\in M$ and time slot $k\in R$ based on their preference for time and subject combinations. Then, solving the complete/partial Umax LSA problem leads to an allocation that maximizes utilitarian social welfare.

\subsubsection*{Balanced Job Rotation in Organizations}
In some companies, employees rotate through different departments or tasks.
Suppose that $n$ employees $N$ rotate between $n$ departments $M$ (e.g., sales, marketing, and operations), over $n$ time periods $R$. 
The Latin square constraint guarantees that:
\begin{itemize}
    \item Each employee works in at most one department at a time.
    \item No department is assigned to multiple employees in each time slot. 
    \item No employee works in the same class more than once.
\end{itemize}
Such a constraint is crucial for balancing workload and cross-training opportunities.
Each employee $i\in N$ has a preference $v_{ijk}$ regarding the ordinal position $k\in R$ for each department $j\in M$. 
Then, this situation can be modeled by the complete Umax LSA problem.

\subsubsection*{Experimental Rotations in Healthcare}
In clinical trials or studies with multiple treatment groups, it is required to design the order of treatment assignments of participants~\cite{chilton1955latin,richardson2018use,preece1991latin}.
For example, in a drug trial with $n$ different treatment methods, each group of patients receives each treatment method exactly once over $n$ treatment periods, ensuring fair treatment exposure.
Suppose that there are $n$ participants groups $N$, $n$ treatments $M$, and $n$ time slots $R$.
Each group of patients $i\in N$ has a preference $v_{ijk}$ regarding the ordinal position $k$ for each treatment $j\in M$. 
Then, the complete Latin square constraint ensures that each patient receives treatments in a structured sequence, ensuring equal administration without repetition in the same order. 
Thus, this situation can be modeled by the complete Umax LSA problem.

\section{Approximation Algorithms for Maximizing Utilitarian Social Welfare}\label{sec:approx-Umax}
In this section, we first present a $(1-1/e)$-approximation algorithm for the partial Umax LSA problem.
This algorithm is based on a technique used for the Latin square extension problem~\cite{gomes2004improved} and the separable assignment problem~\cite{fleischer2006tight}.
The algorithm constructs a configuration LP and then solves it by the \emph{ellipsoid method}~\cite{GLS2012}. It then rounds the solution with a contention resolution scheme.

We then provide a $(1-1/e)/4$-approximation algorithm for the complete Umax LSA problem by showing that $\alpha$-approximate solution of the partial LSA problem can be converted into $\alpha/4$-approximate solution of the complete LSA problem.

\subsection{Partial LSA}
Let $\cS\subseteq 2^{M\times R}$ be the set of all possible (partial) allocations for an agent, i.e., for every $S\in\cS$, $(j,k),(j,k')\in S$ implies $k=k'$ and $(j,k),(j',k)\in S$ implies $j=j'$.
For each agent $i\in N$ and $S\in\cS$, we introduce a binary variable $y_{iS}$, which indicates that $i$ receives $S$.
Then, we reformulate ILP~\eqref{eq:partialILP} as the following (exponential-size) ILP:
\begin{align}
\begin{array}{rll}
\max&\sum_{i\in N} \sum_{S\in\cS}(\sum_{(j,k)\in S}v_{ijk})y_{iS}&\\
\text{s.t.}& \sum_{S\in\cS}y_{iS}=1&(\forall i\in N),\\
& \sum_{i\in N} \sum_{S\in\cS:\, (j,k)\in S}y_{iS}\le 1 &(\forall j\in M,\,\forall k\in R),\\
& y_{iS}\in \{0,1\}  & (\forall i\in N,\,\forall S\in\cS).
\end{array}\label{eq:partialILP2}
\end{align}
A linear programming relaxation of \eqref{eq:partialILP2} is given as
\begin{align}
\begin{array}{rll}
\max&\sum_{i\in N} \sum_{S\in\cS}(\sum_{(j,k)\in S}v_{ijk})y_{iS}&\\
\text{s.t.}& \sum_{S\in\cS}y_{iS}=1&(\forall i\in N),\\
& \sum_{i\in N} \sum_{S\in\cS:\, (j,k)\in S}y_{iS}\le 1 &(\forall j\in M,\,\forall k\in R),\\
& y_{iS}\ge 0  & (\forall i\in N,\,\forall S\in\cS).
\end{array}\label{eq:LP-primal}
\end{align}
Note that the optimum value of \eqref{eq:LP-primal} is an upper bound of the partial Latin square allocation problem.

In what follows, we first show that LP~\eqref{eq:LP-primal} can be solved in a polynomial time by using the ellipsoid method~\cite{GLS2012}. 
This method works when we have a separation algorithm to solve the separation problem for the feasible region.
For a polyhedron $P\subseteq \mathbb{R}^N$, the separation problem for $P$ receives a vector $y$ and either asserts $y\in P$ or finds a vector $d$ such that $d^\top x > d^\top y$ for all $x\in P$.

\begin{lemma}\label{lem:partial-ellipsoid}
There exists a polynomial-time algorithm to solve \eqref{eq:LP-primal}.
\end{lemma}
\begin{proof}
As the number of variables in \eqref{eq:LP-primal} is exponential, we solve it via the following dual:
\begin{align}
{\setlength{\arraycolsep}{3pt}
\begin{array}{rll}
\min&\sum_{i\in N}p_i+\sum_{j\in M}\sum_{k\in R}q_{jk}&\\
\text{s.t.}& p_i+\sum_{(j,k)\in S}q_{jk}\ge \sum_{(j,k)\in S}v_{ijk}&(\forall i\in N,\,\forall S\in\cS),\\
& q_{jk}\ge 0  & (\forall j\in M,\,\forall k\in R).
\end{array}\label{eq:LP-dual}
}
\end{align}
This LP includes an exponential number of constraints but contains only a polynomial number of variables.

For this LP, we construct a separation algorithm.
For a given $(p,q)\in\mathbb{R}^N\times \mathbb{R}_+^{M\times R}$, its feasibility can be checked by computing $\max_{S\in\cS}\sum_{(j,k)\in S}(v_{ijk}-q_{jk})$ for each $i\in N$.
As the problem of maximizing $\sum_{(j,k)\in S}(v_{ijk}-q_{jk})$ can be viewed as a maximum weight matching problem in the complete bipartite graph $(M,R; M\times R)$, this can be solved in a polynomial time (see, e.g., \cite{korte2018combinatorial}).
Thus, we can construct a separation algorithm for LP~\eqref{eq:LP-dual}.

Using the ellipsoid method with this separation algorithm, we can solve LP~\eqref{eq:LP-dual} in a polynomial time~\cite{GLS2012}.
We consider the set of violated inequalities returned by the separation algorithm during the execution of the ellipsoid method. 
Then, the number of these inequalities is polynomial since the ellipsoid method runs in polynomial time.
Additionally, the system of these inequalities is equivalent to that of LP~\eqref{eq:LP-dual}.
Thus, taking the dual of this polynomial-sized dual program results in a primal program with a polynomial number of variables and constraints. 
This can be solved in polynomial time, yielding an optimum solution of \eqref{eq:LP-primal}.
\end{proof}

We then provide a rounded solution that is $(1-1/e)$-approximation for the partial Latin square allocation problem.
Let $y^*$ be an optimum solution of \eqref{eq:LP-primal}.
For each agent $i\in N$, we interpret $(y^*_{iS})_{S\in\cS}$ as probabilities and independently select exactly one allocation $S_i\in \cS$ according to these probabilities.
Each $(j,k)$ is allocated to an agent in $\argmax_{i\in N}\{v_{ijk}\mid (j,k)\in S_i\}$, if the domain of argmax is non-empty.
Our algorithm is summarized as \Cref{alg:partial}.

\begin{algorithm}[ht]
    \caption{$(1-1/e)$-approximation algorithm for the partial case}\label{alg:partial}
    \SetKwInOut{Input}{input}\Input{$(v_{ijk})_{i\in N,\, j\in M,\, k\in R}$}
    \SetKwInOut{Output}{output}\Output{a partial allocation $A$}
    Solve \eqref{eq:LP-primal} and let $y^*$ be an optimum solution of it\;
    For each $i\in N$, we interpret $(y^*_{iS})_{S\in\cS}$ as probabilities and independently select exactly one allocation $S_i\in \cS$ according to these probabilities\;
    Let $A\gets\emptyset$\;
    \ForEach{$j\in M$}{
        \ForEach{$k\in R$}{
            \If{$(j,k)\in S_i$ for some $i$}{
                Let $i^*\in\argmax_{i\in N:\, (j,k)\in S_i}v_{ijk}$\;
                $A\gets A\cup\{(i^*,j,k)\}$\;
            }
        }
    }
    \Return the partial allocation $A$\;
\end{algorithm}

\begin{lemma}\label{lem:partial-round}
The expected utilitarian social welfare of the rounded allocation $A$ of \cref{alg:partial} is at least $(1-1/e)\cdot\sum_{i\in N} \sum_{S\in\cS}(\sum_{(j,k)\in S}v_{ijk})y^*_{iS}$.
\end{lemma}
\begin{proof}
Fix $(j,k)\in M\times R$. 
For each $i\in N$, define $x^*_{ijk}=\sum_{S\in\cS:\,(j,k)\in S}y^*_{iS}$.
Let $\sigma$ be a nonincreasing ordering of the agents with respect to the values $v_{ijk}$, i.e., $v_{\sigma(1)jk}\ge v_{\sigma(2)jk}\ge\dots\ge v_{\sigma(n)jk}$.
In addition, we assume that $\sigma(\ell)<\sigma(\ell+1)$ if $v_{\sigma(\ell)jk}=v_{\sigma(\ell+1)jk}$.
For notational simplicity, we write $x^*_{\sigma(n+1)jk}=1-\sum_{\ell=1}^n x^*_{\sigma(\ell)jk}~(=1-\sum_{i\in N}\sum_{S\in\cS:\,(j,k)\in S}y^*_{iS}\ge 0)$ and $v_{\sigma(n+1)jk}=0$.

By the definition of \cref{alg:partial}, we have $(\sigma(1),j,k)\in A$ with probability $x^*_{\sigma(1)jk}$ and $(\sigma(2),j,k)\in A$ with probability $(1-x^*_{\sigma(1)jk})x^*_{\sigma(2)jk}$.
Similarly, for each $\ell\in\{1,2,\dots,n\}$, we have $(\sigma(\ell),j,k)\in A$ with probability $x^*_{\sigma(\ell)jk}\prod_{t=1}^{\ell-1}(1-x^*_{\sigma(t)jk})$.
Hence, the contribution of $(j,k)$ in the expected value of $A$ is $\sum_{\ell=1}^n v_{\sigma(\ell)jk}\cdot x^*_{\sigma(\ell)jk}\prod_{t=1}^{\ell-1}(1-x^*_{\sigma(t)jk})$.

By using the Chebyshev's sum inequality\footnotemark{} and the AM-GM inequality, we obtain
\begin{align*}
\MoveEqLeft
\sum_{\ell=1}^n v_{\sigma(\ell)jk}\cdot x^*_{\sigma(\ell)jk}\prod_{t=1}^{\ell-1}(1-x^*_{\sigma(t)jk})\\
&=\sum_{\ell=1}^{n+1} x^*_{\sigma(\ell)jk}\cdot v_{\sigma(\ell)jk}\prod_{t=1}^{\ell-1}(1-x^*_{\sigma(t)jk})\\
&\ge \left(\sum_{\ell=1}^{n+1} x^*_{\sigma(\ell)jk}v_{\sigma(\ell)jk}\right)\left(\sum_{\ell=1}^{n+1}x^*_{\sigma(\ell)jk}\prod_{t=1}^{\ell-1}(1-x^*_{\sigma(t)jk})\right)\\
&= \left(\sum_{i\in N} x^*_{ijk}v_{ijk}\right)\left(1-\prod_{t=1}^{n+1}(1-x^*_{\sigma(t)jk})\right)\\
&\ge \left(\sum_{i\in N} x^*_{ijk}v_{ijk}\right)\left(1-\left(\frac{1}{n+1}\sum_{t=1}^{n+1}(1-x^*_{\sigma(t)jk})\right)^{n+1}\right)\\
&=\left(\sum_{i\in N} x^*_{ijk}v_{ijk}\right)\left(1-\left(1-\frac{1}{n+1}\right)^{n+1}\right)\\
&\ge \left(\sum_{i\in N} x^*_{ijk}v_{ijk}\right)\left(1-\frac{1}{e}\right).
\end{align*}

\footnotetext{$\mathbb{E}[f(X)g(X)]\ge \mathbb{E}[f(X)]\mathbb{E}[g(X)]$ if $f$ and $g$ are nonincreasing function~(see, e.g., \cite{besenyei2018picard}).}

Furthermore, we have
\begin{align*}
\sum_{(j,k)\in M\times R}\sum_{i\in N} x^*_{ijk}v_{ijk}
&=\sum_{(j,k)\in M\times R}\sum_{i\in N} v_{ijk}\sum_{S\in\cS:\, (j,k)\in S}y^*_{iS}\\
&=\sum_{i\in N}\sum_{S\in\cS}\bigg(\sum_{(j,k)\in S}v_{ijk}\bigg)y^*_{iS}.
\end{align*}
Hence, the expected utilitarian social welfare of the rounded allocation $A$ is at least
\begin{align*}
\sum_{(j,k)\in M\times R}\bigg(\sum_{i\in N} x^*_{ijk}v_{ijk}\bigg)\bigg(1-\frac{1}{e}\bigg)
&\ge \left(1-\frac{1}{e}\right)\sum_{i\in N}\sum_{S\in\cS}\bigg(\sum_{(j,k)\in S}v_{ijk}\bigg)y^*_{iS},
\end{align*}
and the proof is complete.
\end{proof}

We get the following theorem from \Cref{lem:partial-ellipsoid,lem:partial-round}.
\begin{theorem}\label{thm:partial-approx}
\cref{alg:partial} is a $(1-1/e)$-approximation algorithm for the partial Umax LSA problem.
\end{theorem}

It is worth mentioning that we can derandomize \cref{alg:partial} in polynomial time by a standard technique using a conditional expectation.
Indeed, by sequentially selecting a matching $S_i\in\{S\in\cS\mid y^*_{iS}>0\}$ that maximizes the conditional expected utilitarian social welfare, we can deterministically obtain a partial allocation whose utilitarian social welfare is at least the expected utilitarian social welfare of \cref{alg:partial}.
Thus, we can conclude that there exists a deterministic $(1-1/e)$-approximation algorithm for the partial Umax LSA problem.\footnote{For more details of this derandomization technique, see, e.g., the book by \citet[Section 5.2]{williamson2011design}.}

\subsection{Complete LSA}
In this subsection, we construct a complete allocation from a partial allocation without reducing utilitarian social welfare by more than a quarter. Our algorithm first divides the given partial allocation into four blocks, such that any of them can be extended to a complete allocation. It then extends the block with the maximum utilitarian social welfare.

The division is based on Ryser's theorem for a Latin rectangle extension~\cite{ryser1951combinatorial}. It implies that any partial allocation $A$ that satisfies the following conditions can be extended to a complete allocation: there exist positive integers $m$ and $r$ such that $m+r\le n$ and $A(j,k)\ne\bot$ if and only if $j\in [m]$ and $k\in [r]$.\footnotemark{}
The proof of Ryser's theorem is constructive, based on K\"{o}nig's edge coloring theorem. 
It is well known that the edge coloring can be found efficiently.
\begin{lemma}[\cite{Schrijver1998}]\label{lem:Konig}
For a bipartite graph with $m$ edges and a maximum degree of $\Delta$, we can find an edge coloring with $\Delta$ colors in $O(m\Delta)$ time.
\end{lemma}

\footnotetext{The actual Ryser's theorem~\cite{ryser1951combinatorial} is the following stronger statement.
Let $M'\subseteq M$ and $R'\subseteq R$ with $|M'|=m$ and $|R'|=r$. 
For any partial allocation $A$ such that $A(j,k)\ne\bot$ if and only if $(j,k)\in M'\times R'$, there exists a complete allocation $\bar{A}\supseteq A$ if and only if
$|\{(j,k)\in M'\times R'\mid A(j,k)=i\}|\ge m+r-n$ for all $i\in N$.}

For a partial allocation $A$, let $M'=\{j\in M\mid A(j,k)\ne\bot~(\exists k\in R)\}$ be the set of items that are not assigned in at least one round, and let $R'=\{k\in R\mid A(j,k)\ne\bot~(\exists j\in M)\}$ be the set of rounds in which at least one item is unassigned.
Suppose that $|M'|+|R'|\le n$.
Then, we construct an extension $\overline{A}$ such that $A\subseteq \overline{A}\subseteq N\times M'\times R'$ by allocating every $(j,k)\in M'\times R'$ with $A(j,k)=\bot$ in a greedy manner without violating feasibility.
Now, we can apply the construction method of Ryser's theorem~\cite{ryser1951combinatorial}.
We first make an extension with respect to rounds.
To do so, we construct a bipartite graph $G_1=(N,M';\{(p,q)\in N\times M'\mid \overline{A}(q,k)\ne p~(\forall k\in R')\})$, where an edge $(p,q)$ means that agent $p$ does not get item $q$ yet.
Then, we compute an edge coloring of $G_1$ with $n-|R'|$ colors. Here, colors correspond to rounds not allocated yet (i.e., $R\setminus R'$). We allocate an item $q$ to agent $p$ according to the color of edge $(p,q)$.
Note that such a coloring of $n-|R'|$ colors exists since degrees of left-side vertices are at most $|M'|\le n-|R'|$ and degrees of right-side vertices are exactly $n-|R'|$.
We then make an extension with respect to items.
We construct a bipartite graph $G_2=(N,R;\{(p,q)\in N\times R\mid \overline{A}(j,q)\ne p~(\forall j\in M')\})$, where an edge $(p,q)$ means that $p$ does not get any item at round $q$.
We compute an edge coloring of $G_2$ with $n-|M'|$ colors and allocate. Here, colors correspond to rounds not allocated yet (i.e., $R\setminus R'$). We allocate an item corresponding to the color of $(p,q)$ to agent $p$ at round $q$.
Note that such a coloring of $n-|M'|$ colors exists since degrees of vertices are exactly $n-|M'|$.
Our algorithm is summarized in \cref{alg:extend}.

\begin{algorithm}[ht]
    \caption{Extending a partial allocation}\label{alg:extend}
    \SetKwInOut{Input}{input}\Input{$(v_{ijk})_{i\in N,\, j\in M,\, k\in R}$ and a partial allocation $A$ such that $|\{j\in M\mid A(j,k)\ne\bot~(\exists k\in R)\}|+|\{k\in R\mid A(j,k)\ne\bot~(\exists j\in M)\}|\le n$}
    \SetKwInOut{Output}{output}\Output{a complete allocation $\overline{A}\supseteq A$}
    Let $M'\gets \{j\in M\mid A(j,k)\ne\bot~(\exists k\in R)\}$ and $m\gets |M'|$\; 
    Let $R'\gets \{k\in R\mid A(j,k)\ne\bot~(\exists j\in M)\}$ and $r\gets |R'|$\;
    Set $\overline{A}\gets A$\;
    \tcc{Fill $M'\times R'$}
    \ForEach{$(j,k)\in M'\times R'$}{
        \If{$A(j,k)=\bot$}{
            Select an $i\in N$ such that $\overline{A}(j',k)\ne i~(\forall j'\in M')$ and $\overline{A}(j,k')\ne i~(\forall k'\in R')$\;
            Set $\overline{A}\gets \overline{A}\cup\{(i,j,k)\}$\;
        }
    }
    \tcc{Fill $M'\times R$}
    Construct a bipartite graph $G_1=(N,M';\{(p,q)\in N\times M'\mid \overline{A}(q,k)\ne p~(\forall k\in R')\})$\;
    Compute an edge coloring of $G_1$ with $n-r$ colors and let $E_1,E_2,\dots,E_{n-r}$ be the partition of the edges corresponds to the coloring\;
    Let $\ell\gets 1$\;
    \ForEach{$k\in R\setminus R'$}{
        \lForEach{$(p,q)\in E_{\ell}$}{
            $\overline{A}\gets \overline{A}\cup\{(p,q,k)\}$
        }
        $\ell\gets \ell+1$\;
    }
    \tcc{Fill $M\times R$}
    Construct a bipartite graph $G_2=(N,R;\{(p,q)\in N\times R\mid \overline{A}(j,q)\ne p~(\forall j\in M')\})$\;
    Compute an edge coloring of $G_2$ with $n-m$ colors and let $F_1,F_2,\dots,F_{n-m}$ be the partition of the edges corresponds to the coloring\;
    Let $\ell\gets 1$\;
    \ForEach{$j\in M\setminus M'$}{
        \lForEach{$(p,q)\in F_{\ell}$}{
            $\overline{A}\gets \overline{A}\cup\{(p,j,q)\}$
        }
        $\ell\gets \ell+1$\;
    }    
    \Return complete allocation $\overline{A}$\;
\end{algorithm}

\begin{lemma}\label{lem:extend}
Given a partial allocation $A$ such that $|\{j\in M\mid A(j,k)\ne\bot~(\exists k\in R)\}|+|\{k\in R\mid A(j,k)\ne\bot~(\exists j\in M)\}|\le n$, the extension of $A$ computed by \cref{alg:extend} is a complete allocation. Moreover, \cref{alg:extend} can be implemented to run in $O(n^3)$ time.
\end{lemma}
\begin{proof}
The allocation computed by \cref{alg:extend} is a complete extension based on the above discussion. The computational time is $O(n^3)$ according to \Cref{lem:Konig}, since $G_1$ and $G_2$ have at most $O(n^2)$ edges and degrees of vertices are at most $n$.
\end{proof}

By combining \Cref{thm:partial-approx} and \Cref{lem:extend}, we provide a $(1-1/e)/4$-approximate algorithm for the complete LSA problem.
Let $A$ be the partial allocation obtained by \cref{alg:partial}.
If $A$ is a complete allocation, then $A$ is a desired allocation.
Otherwise (i.e., $A(j,k)=\bot$ for some $(j,k)\in M\times R$), we partition it into four equal-sized blocks.
If $n$ is even, the sizes of the blocks are $n/2\times n/2$.
If $n$ is odd, the sizes of the blocks are $(n+1)/2 \times (n-1)/2$ or $(n-1)/2\times (n+1)/2$, where one cell is remaining. We set that the remaining cell is unassigned. 
We choose a block with a maximum utilitarian social welfare and extend it to a complete allocation by \cref{alg:extend}.
Our algorithm is formally described in \cref{alg:complete}.

\begin{algorithm}[t]
    \caption{$(1-1/e)/4$-approximation algorithm for the complete case}\label{alg:complete}
    \SetKwInOut{Input}{input}\Input{$(v_{ijk})_{i\in N,\, j\in M,\, k\in R}$}
    \SetKwInOut{Output}{output}\Output{a complete allocation}
    Compute a partial allocation $A$ by \cref{alg:partial}\;
    \lIf{$A$ is complete }{\Return $A$}
    \tcc{Divide $A$ into four blocks}
    \If{$n$ is even}{
        $A^{(1)}\gets \{(i,j,k)\in A\mid j\le n/2,\ k\le n/2\}$\;
        $A^{(2)}\gets \{(i,j,k)\in A\mid j\le n/2,\ k> n/2\}$\;
        $A^{(3)}\gets \{(i,j,k)\in A\mid j>n/2,\ k\le n/2\}$\;
        $A^{(4)}\gets \{(i,j,k)\in A\mid j>n/2,\ k>n/2\}$\;
    }
    \Else{
        Relabel the items and the rounds so that $A((n+1)/2,(n+1)/2)=\bot$\;
        $A^{(1)}\gets \{(i,j,k)\in A\mid j< (n+1)/2,\ k\le (n+1)/2\}$\;
        $A^{(2)}\gets \{(i,j,k)\in A\mid j\le (n+1)/2,\ k> (n+1)/2\}$\;
        $A^{(3)}\gets \{(i,j,k)\in A\mid j\ge (n+1)/2,\ k< (n+1)/2\}$\;
        $A^{(4)}\gets \{(i,j,k)\in A\mid j> (n+1)/2,\ k\ge(n+1)/2\}$\;
    }
    Let $A^*\gets\argmax_{A^{(\ell)}}\sum_{(i,j,k)\in A^{(\ell)}}v_{ijk}$\;
    \Return extension of $A^*$ computed by \cref{alg:extend}\;
\end{algorithm}


\begin{theorem}\label{thm:complete-approx}
\cref{alg:complete} is a $(1-1/e)/4$-approximation algorithm for the complete Umax LSA.
\end{theorem}
\begin{proof}
Let $\OPT$ and $\overline{\OPT}$ be the optimum values of \eqref{eq:partialILP} and \eqref{eq:completeILP}, respectively. 
Additionally, let $A$ and $A^*$ be the partial and complete allocations obtained by \cref{alg:partial} and \cref{alg:complete}.
Then, we have $\overline{\OPT}\le \OPT\le \frac{1}{1-1/e}\cdot\sum_{(i,j,k)\in A}v_{ijk}\le \frac{4}{1-1/e}\cdot \sum_{(i,j,k)\in\overline{A}}v_{ijk}$ by \Cref{thm:partial-approx}.
This means that \cref{alg:complete} is a $(1-1/e)/4$-approximation algorithm.
\end{proof}

\section{FPT algorithms for Maximizing Utilitarian Social Welfare}\label{sec:FPT-Umax}
 
In this section, we provide FPT algorithms for the Umax LSA problems with respect to the number of agents $n$ and the optimum value, respectively.

Regarding the value of $n$, the task is not difficult because we can enumerate all the possible allocations in FPT time.
More precisely, the number of complete and partial allocations are at most $n^{n^2}\le 2^{O(n^2\log n)}$ and $(n+1)^{n^2}\le 2^{O(n^2\log n)}$ by considering all the possible assignments of each cell.\footnotemark 
\footnotetext{Let $L(n)$ be the number of complete Latin squares of order $n$. It is known that $(L(n))^{1/n^2}\sim e^{-2}n$~\cite[Theorem 17.3]{van2001course}, and hence $L(n)=2^{\Theta(n^2\log n)}$.}
For each allocation, we can check the feasibility and compute the utilitarian social welfare in $O(n^2)$ time. 
It should be noted that we can also solve the Emax problem in the same manner.
Thus, we obtain the following theorem.
\begin{theorem}\label{thm:FPT-n}
There are FPT algorithms with respect to $n$ whose computational complexity is $e^{O(n^2\log n)}$ for the LSA problems of partial/complete Umax/Emax.
\end{theorem}

When the optimum value is the parameter, we use the \emph{color coding technique}~\cite{cygan2015parameterized} to construct FPT algorithms.
Let $A^*$ be an optimum solution and $\alpha=\sum_{(i,j,k)\in A^*}v_{ijk}$.
If $\alpha\ge n/2$, the FPT algorithms with respect to $n$ are applicable.
Hence, without loss of generality, we may assume that $\alpha<n/2$.
Let $S^*=\{(j,k)\in M\times R\mid (i,j,k)\in A^* \text{ and } v_{ijk}>0\}$ be the set of item--round pairs that are assigned to some agent in the optimum solution with positive utility, and 
let $T^*=\{i\in N\mid (i,j,k)\in A^* \text{ and } v_{ijk}>0\}$ be the set of agents who receive positive in the optimum solution.
Additionally, let $s^*=|S^*|$ and $t^*=|T^*|$.
Since $v_{ijk}\ge 1$ for any $(i,j,k)\in S^*$, it follows that $t^*\le s^*\le \alpha$.
We guess the values of $s^*$ and $t^*$, and denote these guessed values as $s$ and $t$, respectively.
The parameter $s$ is inferred sequentially as $1,2,\dots$.
If the estimated $s$ is at most $\alpha$, then we can find an allocation whose utilitarian social welfare is at least $s$ for some $t\in [s]$. Therefore, if the optimal utilitarian social welfare currently obtained is less than $s$, it can be concluded that $\alpha$ is less than $s$, and the process can be terminated.


Let $\chi\colon M\times R\to [s]$ be a coloring of cells $M\times R$ and let $\psi\colon [s]\to [t]$ be a map from $[s]$ to $[t]$.
We will select a coloring such that the elements in $S$ are colored with pairwise distinct colors.
Additionally, we will select $\psi$ such that 
$\psi\circ\chi(j,k)=\psi\circ\chi(j',k')$ if and only if $i=i'$ for all $(i,j,k),(i',j',k')\in A^*$ with $v_{ijk},v_{i'j'k'}>0$.

We enumerate all possible maps to select a desired function of $\psi$.
Let $\mathcal{P}$ be the set of all possible functions from $[s]$ to $[t]$, which is at most $t^s\le \alpha^\alpha$.
For coloring $\chi$, we use a tool called \emph{splitter}.
An $(a,b,c)$-splitter is a family $\mathcal{F}$ of functions from $[a]$ to $[c]$ such that, for any $X\subseteq [a]$ of size $b$, there exists $f\in\mathcal{F}$ that is injective on $X$.
It is known that there exists an $(a,b,b)$-splitter of size $e^{O(b\log^2 b)}\cdot\log a$ that can be constructed in time $e^{O(b\log^2 b)}\cdot a\log a$~\cite{naor1995splitters,cygan2015parameterized}.
Let $\mathcal{X}$ be a $(n^2, s, s)$-splitter.

Suppose that $\chi\in\mathcal{X}$ and $\psi\in\mathcal{P}$ satisfy the desired condition: $\psi\circ\chi(j,k)=\psi\circ\chi(j',k')$ if and only if $i=i'$ for all $(i,j,k),(i',j',k')\in A^*$ with $v_{ijk},v_{i'j'k'}>0$.
For each $\ell\in [t]$, let $\mathcal{Q}_\ell\subseteq \{(j,k)\in M\times R\mid \psi\circ\chi(j,k)=\ell\}$ be the set of all possible (feasible) allocations for an agent.
For each $i\in N$ and $\ell\in[t]$, let 
\begin{align}
Q_{i\ell}\in\argmax\left\{v_i(Q)\mid Q\in\mathcal{Q}_\ell\right\}. \label{eq:Qil}
\end{align}
Then, we can find an optimum partial allocation by computing the maximum weight matching on a complete bipartite graph between $N$ and $[t]$ with weights $v_i(Q_{i\ell})$ for each $(i,\ell)\in N\times[t]$.

Our algorithm outputs the best assignment among those found by the above procedure. 
If the optimum value is less than $n/2$, it outputs a complete allocation by using \cref{alg:extend}.
The algorithm is formally described in \cref{alg:FPT}.

\begin{algorithm}[ht]
    \caption{FPT algorithm}\label{alg:FPT}
    \SetKwInOut{Input}{input}\Input{$(v_{ijk})_{i\in N,\, j\in M,\, k\in R}$}
    \SetKwInOut{Output}{output}\Output{an optimal allocation}
    Let $A^*$ be an arbitrary complete allocation and $u\gets 0$\;
    \For{$s\gets 1,2,\dots$}{
        \For{$t\gets 1,2,\dots,s$}{
            Let $\mathcal{P}$ be the set of all possible functions from $[s]$ to $[t]$\;
            Let $\mathcal{X}$ be an $(n^2,s,s)$-splitter\;
            \ForEach{$(\psi,\chi)\in \mathcal{P}\times\mathcal{X}$}{
                Define $Q_{i\ell}$ for each $i\in N$ and $\ell\in[t]$ as in \eqref{eq:Qil}\;
                Construct a complete bipartite graph between $N$ and $[t]$ with weights $(v_i(Q_{i\ell}))_{i\in N,\,\ell\in[t]}$\;
                Compute a maximum weight matching $\mu\subseteq N\times [t]$ in the bipartite graph and let $v$ be its weight\;
                \If{$v>u$}{
                    $u\gets v$ and
                    $A^*\gets \{(i,j,k)\mid (i,\ell)\in\mu,\ (j,k)\in Q_{i\ell},\ v_{ijk}>0\}$\;
                }
            }
        }
        \If{$u\ge n/2$}{
            Enumerate all the possible allocations and \Return the optimum one\;
        }
        \lIf{$s>u$}{\Return extension of $A^*$ computed by \cref{alg:extend}}
    }
\end{algorithm}

\begin{theorem}\label{thm:FPT-alpha}
There exist FPT algorithms with respect to the Umax value for both the partial and complete LSA problems.
\end{theorem}
\begin{proof}
It is sufficient to prove that the time complexity of \cref{alg:FPT} is FPT with respect to $\alpha$.
The computational complexity of each iteration for $s$ and $t$ is at most 
\[e^{O(s\log^2 s)}\cdot n^2\log n + e^{O(s\log^2 s)}\cdot O(n^3) = e^{O(s\log^2 s)}\cdot n^3.\]
If $\alpha<n/2$, then the total computational time is at most $\alpha^2\cdot e^{O(\alpha\log^2 \alpha)}\cdot n^3+O(n^3)$ by \Cref{lem:extend}, which is FPT.
If $\alpha\ge n/2$, then the total computational time is at most 
\[(n/2)^2\cdot e^{O(n\log^2 n)}\cdot n^3+e^{O(n^2\log n)}=e^{O(\alpha^2\log \alpha)}\] 
by \Cref{thm:FPT-n}, which is also FPT.
\end{proof}

\section{NP-hardness}\label{sec:NP-hard}
In this section, we present various results on the NP-hardness of finding desirable allocations.
Due to space limitations, some proofs are deferred to the appendix.

We begin by addressing the hardness of the binary case.
\begin{theorem}\label{thm:Umax-hard}
When the valuations are binary, deciding whether there exists a complete allocation $A$ such that $v_{ijk}=1$ for all $(i,j,k)\in A$ is strongly NP-complete.
\end{theorem}
\begin{proof}
It is clear that the problem is in the class NP.
We present a polynomial-time reduction from a partial Latin square problem, which is known to be strongly NP-complete~\cite{colbourn1984complexity}.
In the partial Latin square completion problem, we are given a partial Latin square $P\subseteq N\times M\times R$, and the goal is to check whether it can be extended to a complete Latin square. 

We construct valuations from the given partial Latin square $P$ as follows:
\begin{align*}
v_{ijk}&=
\begin{cases}
    0 & \text{if }P(j,k)\in N\setminus\{i\},\\
    1 & \text{if }P(j,k)\in \{i,\bot\}
\end{cases}&(i\in N,\ j\in M,\ k\in R).
\end{align*}
Clearly, the valuations are binary.

Suppose that $P$ can be extended to a complete Latin square $\overline{P}\supseteq P$.
Then, by interpreting $\overline{P}$ as a complete allocation, we have $v_{ijk}=1$ for all $(i,j,k)\in \overline{P}$ because $\overline{P}(j,k)=i$ implies $P(j,k)\in\{i,\bot\}$.

Conversely, suppose that there exists a complete allocation $A$ such that $v_{ijk}=1$ for all $(i,j,k)\in A$. Then, by the definition of the valuation, we have $A(j,k)=i$ if $P(j,k)=i$. This means that $A$ is a complete Latin square that extends $P$.

Therefore, determining whether a complete allocation $A$ with $v_{ijk}=1$ for all $(i,j,k)\in A$ exists is strongly NP-complete.
\end{proof}

From this theorem, we can conclude that computing a Umax or Emax allocation is NP-hard, even when the valuations are binary.
Indeed, for the binary case, deciding whether there exists an allocation whose utilitarian social welfare is $n^2$ is NP-complete, and deciding whether there exists an allocation whose egalitarian social welfare is $n$ is NP-complete.

Moreover, we can also conclude that computing a \emph{non-wastefulness} allocation and a \emph{Pareto optimal} allocation are both NP-hard.
Here, an allocation $A$ is said to be non-wasteful if, for each $(j,k)\in M\times R$ where $v_{i'jk}>0$ for some $i'\in N$, there exists an agent $i\in N$ such that $v_{ijk}>0$ and $A(j,k)=i$.
Additionally, an allocation $A$ is said to be \emph{Pareto optimal} if there exists no allocation $B$ such that $v_i(B_i)\ge v_i(A_i)$ for every agent $i\in N$, with at least one of the inequalities being strict.
\begin{corollary}\label{cor:PO}
Even when the valuations are binary, computing a partial or complete allocation that satisfies each of Umax, Emax, non-wastefulness, or Pareto optimal is strongly NP-hard.
\end{corollary}

Next, we show that deciding whether the Emax value is $1$ or $2$ is NP-hard.
This suggests that there is no FPT algorithm for the Emax LSA problem with respect to the optimum value.
\begin{theorem}\label{thm:Emax-hard}
    Even when the valuations are binary, deciding whether the Emax value is $1$ or $2$ is strongly NP-hard for both the partial and complete LSA problems.
Furthermore, deciding whether the Umax value is at least $2n$ or less is strongly NP-hard.
\end{theorem}
\begin{proof}
We first consider the partial setting.
We reduce from \emph{4-occurrence-3SAT} (2L-OCC-3SAT). This version of 3SAT is where each literal, both positive and negative, occurs exactly twice in the clauses. Thus, each variable occurs four times in the clauses. 2L-OCC-3SAT, and even its monotone version, are known to be NP-hard \cite{darmann2021simplified}. 

We consider such a 3SAT formula $\theta$ on $\lambda$ variables $x_1, \ldots, x_\lambda$ and $\mu$ clauses $C_1, \ldots, C_\mu$. We construct an instance of the partial LSA problem as follows. Note that $3\mu=4\lambda$.

We construct three kinds of agents: variable agents, transfer agents, and clause agents. 
For each variable $x_k$ (or, clause $C_l$), we have a variable agent (respectively, clause agent) with the same name. Furthermore, for each variable $x_k$, we have four transfer agents $t_k^1,t_k^2,t_k^3,t_k^4$.
Thus, $N=\{C_1,\dots,C_\mu\}\cup\bigcup_{k\in[\lambda]}\{x_k,t_k^1,t_k^2,t_k^3,t_k^4\}$ and $n=\mu+5\lambda=19\lambda/3$.
Let $M=R=[n]$.
We write $a(j,k)$ to denote the valuation of cell $(j,k)\in M\times R$ for agent $a\in N$. 

For each $k\in[\lambda]$, we set $x_k(2k-1, 2k-1) = x_k(2k-1, 2k) = x_k(2k, 2k-1) = x_k(2k, 2k) = 1$, and all other cells are valued at $0$ for $x_k$.
We also set $t_k^1(2k-1, 2k-1) = t_k^2(2k-1, 2k) = t_k^3(2k, 2k) = t_k^4(2k, 2k-1) = 1$. Suppose that the positive literal of $x_k$ occurs in the clauses $C_l$ and $C_p$, where $l<p$. Then we set $t_k^1(2\lambda + l, 2k-1) = t_k^3(2\lambda + p, 2k) = 1$. Suppose that the negative literal of $x_k$ occurs in the clauses 
$C_l$ and $C_p$, where $l<p$. 
Then we set $t_k^2(2\lambda + l, 2k) = t_k^4(2\lambda + p, 2k-1) = 1$. 
Such cells that occur below the $2\lambda^{th}$ row and give value $1$ for some transfer agent, we call the \emph{bottom cells} of their respective transfer agents.
We also set $t_k^v(1, 4k - 4 + v)= 1$ for each transfer agent $t_k^v$, for each $k>1$. 
Note that these cells in the first row do not share the same rows or columns as the other cells that give a value of $1$ for a transfer agent. Thus, these can be allocated to the transfer agents irrespective of the other valued cells that they get.
We set ${t_1^1(2,3)= t_1^2(2,4)=t_1^3(1,3)=t_1^4(1,4)}=1$. We call these the \emph{top} cells of the transfer agents.
All the other cells are evaluated at $0$ for the transfer agents.
\begin{table*}[htbp]
\centering
\caption{Reduced instance of the monotone 4-occurrence 3SAT problem {$(x_1\vee x_2\vee x_3)
\wedge (x_1\vee x_2\vee x_4)  
\wedge (x_3\vee x_5\vee x_6)
\wedge (x_4\vee x_5\vee x_6)
\wedge
(\overline{x}_1\vee \overline{x}_2\vee \overline{x}_3)
\wedge (\overline{x}_1\vee \overline{x}_2\vee \overline{x}_4)  
\wedge (\overline{x}_3\vee \overline{x}_5\vee \overline{x}_6)
\wedge (\overline{x}_4\vee \overline{x}_5\vee \overline{x}_6)
$}. 
Each cell displays the agents that evaluate it to $1$. If it is $\emptyset$, then the cell is evaluated to $0$ for all agents. 
The allocation represented by red color corresponds to a truth assignment of $(x_1,x_2,x_3,x_4,x_5,x_6)=(\true,\false,\true, \true, \true, \false)$.}\label{tab:fixed-EF1-hard}
\tabcolsep = 2mm
\renewcommand{\arraystretch}{1.4}
\newcommand{\myline}{\hhline{~---------------------}}
\scalebox{.55}{
\begin{tabular}{S|c|c|c|c|c|c|c|c|c|c|c|c|c|c|c|c|c|c|c|c|c|}
\noline{}            &\noline{1}   & \noline{2} & \noline{3}  & \noline{4} & \noline{5} &\noline{6}& \noline{7} & \noline{8} & \noline{9} & \noline{10} & \noline{11}   & \noline{12} & \noline{13} & \noline{$\cdots$} & \noline{24}& \noline{25} & \noline{$\cdots$} & \noline{32} & \noline{33} & \noline{$\cdots$} & \noline{38} \\\myline
1       &$x_1,\alert{t_1^1}$& $\alert{x_1},t_1^2$     &$\alert{t_1^3}$& $\alert{t_1^4}$   
&$\alert{t_2^1}$& $\alert{t_2^2}$    &$\alert{t_2^3}$& $\alert{t_2^4}$    &$\alert{t_3^1}$& $\alert{t_3^2}$ &$\alert{t_3^3}$& $\alert{t_3^4}$&$\alert{t_4^1}$&$\alert{\cdots}$&$\alert{t_6^4}$&$\alert{C_1}$&$\alert{\cdots}$&$\alert{C_8}$&$\emptyset$&$\cdots$&$\emptyset$\\\myline
2       &$\alert{x_1}, t_1^4$& $x_1,\alert{t_1^3}$     &$\alert{t_1^1}$& $\alert{t_1^2}$     &$\emptyset$& $\emptyset$     &$\emptyset$& $\emptyset$    &$\emptyset$& $\emptyset$ &$\emptyset$& $\emptyset$ &$\emptyset$&$\cdots$&$\emptyset$&$\emptyset$&$\cdots$&$\emptyset$&$\emptyset$&$\cdots$&$\emptyset$\\\myline
3      &$\emptyset$& $\emptyset$     &$\alert{x_2},t_2^1$& $x_2,\alert{t_2^2}$     &$\emptyset$& $\emptyset$    &$\emptyset$& $\emptyset$    &$\emptyset$& $\emptyset$ &$\emptyset$& $\emptyset$&$\emptyset$&$\cdots$&$\emptyset$&$\emptyset$&$\cdots$&$\emptyset$&$\emptyset$&$\cdots$&$\emptyset$\\\myline
4      &$\emptyset$& $\emptyset$     &$x_2,\alert{t_2^4}$& $\alert{x_2},t_2^3$     &$\emptyset$& $\emptyset$     &$\emptyset$& $\emptyset$   &$\emptyset$& $\emptyset$ &$\emptyset$& $\emptyset$&$\emptyset$&$\cdots$&$\emptyset$&$\emptyset$&$\cdots$&$\emptyset$&$\emptyset$&$\cdots$&$\emptyset$\\\myline
5 &$\emptyset$& $\emptyset$     &$\emptyset$& $\emptyset$     &$x_3,\alert{t_3^1}$& $\alert{x_3},t_3^2$     &$\emptyset$& $\emptyset$    &$\emptyset$& $\emptyset$ &$\emptyset$& $\emptyset$&$\emptyset$&$\cdots$&$\emptyset$&$\emptyset$&$\cdots$&$\emptyset$&$\emptyset$&$\cdots$&$\emptyset$\\\myline
6       &$\emptyset$& $\emptyset$    & $\emptyset$& $\emptyset$     &$\alert{x_3},t_3^4$& $x_3,\alert{t_3^3}$    &$\emptyset$& $\emptyset$    &$\emptyset$& $\emptyset$ &$\emptyset$& $\emptyset$&$\emptyset$&$\cdots$&$\emptyset$&$\emptyset$&$\cdots$&$\emptyset$&$\emptyset$&$\cdots$&$\emptyset$\\\myline
7      &$\emptyset$& $\emptyset$    &$\emptyset$& $\emptyset$     &$\emptyset$& $\emptyset$     &$x_4,\alert{t_4^1}$& $\alert{x_4},t_4^2$    &$\emptyset$& $\emptyset$ &$\emptyset$& $\emptyset$&$\emptyset$&$\cdots$&$\emptyset$&$\emptyset$&$\cdots$&$\emptyset$&$\emptyset$&$\cdots$&$\emptyset$\\\myline
8 &$\emptyset$& $\emptyset$     &$\emptyset$& $\emptyset$     &$\emptyset$& $\emptyset$     &$\alert{x_4},t_4^4$& $x_4,\alert{t_4^3}$    &$\emptyset$& $\emptyset$ &$\emptyset$& $\emptyset$&$\emptyset$&$\cdots$&$\emptyset$&$\emptyset$&$\cdots$&$\emptyset$&$\emptyset$&$\cdots$&$\emptyset$\\\myline
9      &$\emptyset$& $\emptyset$     &$\emptyset$& $\emptyset$     &$\emptyset$& $\emptyset$    &$\emptyset$& $\emptyset$   &$x_5,\alert{t_5^1}$& $\alert{x_5},t_5^2$ &$\emptyset$& $\emptyset$&$\emptyset$&$\cdots$&$\emptyset$&$\emptyset$&$\cdots$&$\emptyset$&$\emptyset$&$\cdots$&$\emptyset$\\\myline
10      &$\emptyset$& $\emptyset$    &$\emptyset$& $\emptyset$     &$\emptyset$& $\emptyset$     &$\emptyset$& $\emptyset$    &$\alert{x_5},t_5^4$& $x_5,\alert{t_5^3}$ &$\emptyset$& $\emptyset$&$\emptyset$&$\cdots$&$\emptyset$&$\emptyset$&$\cdots$&$\emptyset$&$\emptyset$&$\cdots$&$\emptyset$\\\myline
11 &$\emptyset$& $\emptyset$     &$\emptyset$& $\emptyset$     &$\emptyset$& $\emptyset$    &$\emptyset$& $\emptyset$    &$\emptyset$& $\emptyset$ &$\alert{x_6},t_6^1$& $x_6,\alert{t_6^2}$&$\emptyset$&$\cdots$&$\emptyset$&$\emptyset$&$\cdots$&$\emptyset$&$\emptyset$&$\cdots$&$\emptyset$\\\myline
12      &$\emptyset$& $\emptyset$    &$\emptyset$& $\emptyset$     &$\emptyset$& $\emptyset$     &$\emptyset$& $\emptyset$    &$\emptyset$& $\emptyset$ &$x_6,\alert{t_6^4}$& $\alert{x_6},t_6^3$&$\emptyset$&$\cdots$&$\emptyset$&$\emptyset$&$\cdots$&$\emptyset$&$\emptyset$&$\cdots$&$\emptyset$\\\myline

13       &$t_1^1, \alert{C_1}$& $\emptyset$     &$\alert{t_2^1}, C_1$& $\emptyset$     & $t_3^1, C_1$& $\emptyset$     &$\emptyset$&$\emptyset$    &$\emptyset$& $\emptyset$ &$\emptyset$& $\emptyset$&$\emptyset$&$\cdots$&$\emptyset$&$\emptyset$&$\cdots$&$\emptyset$&$\emptyset$&$\cdots$&$\emptyset$\\\myline
14 &$\emptyset$& $t_1^3, \alert{C_2}$     &$\emptyset$& $\alert{t_2^3}, C_2$     &$\emptyset$& $\emptyset$     &$t_4^1, C_2$& $\emptyset$    &$\emptyset$& $\emptyset$ &$\emptyset$& $\emptyset$&$\emptyset$&$\cdots$&$\emptyset$&$\emptyset$&$\cdots$&$\emptyset$&$\emptyset$&$\cdots$&$\emptyset$\\\myline
15       &$\emptyset$& $\emptyset$     &$\emptyset$& $\emptyset$     & $\emptyset$ & $t_3^3, C_3$     &$\emptyset$& $\emptyset$   &$t_5^1, \alert{C_3}$& $\emptyset$ &$\alert{t_6^1}, C_3$& $\emptyset$&$\emptyset$&$\cdots$&$\emptyset$&$\emptyset$&$\cdots$&$\emptyset$&$\emptyset$&$\cdots$&$\emptyset$\\\myline
16 & $\emptyset$    & $\emptyset$     & $\emptyset$ & $\emptyset$     & $\emptyset$ & $\emptyset$     & $\emptyset$   & $t_4^3, C_4$      & $\emptyset$  & $t_5^3, \alert{C_4}$   & $\emptyset$  & $\alert{t_6^3}, C_4$&$\emptyset$&$\cdots$&$\emptyset$&$\emptyset$&$\cdots$&$\emptyset$&$\emptyset$&$\cdots$&$\emptyset$\\\myline

17    & $\emptyset$   &$\alert{t_1^2}, C_5$ &$\emptyset$    &$t_2^2, \alert{C_5}$& $\emptyset$     & $\alert{t_3^2}, C_5$& $\emptyset$&$\emptyset$    &$\emptyset$& $\emptyset$ &$\emptyset$& $\emptyset$&$\emptyset$&$\cdots$&$\emptyset$&$\emptyset$&$\cdots$&$\emptyset$&$\emptyset$&$\cdots$&$\emptyset$\\\myline
18& $\alert{t_1^4}, C_6$     &$\emptyset$& $t_2^4, \alert{C_6}$ &$\emptyset$&$\emptyset$&$\emptyset$ &$\emptyset$         &$\alert{t_4^2}, C_6$& $\emptyset$     &$\emptyset$& $\emptyset$ & $\emptyset$&$\emptyset$&$\cdots$&$\emptyset$&$\emptyset$&$\cdots$&$\emptyset$&$\emptyset$&$\cdots$&$\emptyset$\\\myline
19     &$\emptyset$& $\emptyset$     &$\emptyset$& $\emptyset$      & $\alert{t_3^4}, C_7$     &$\emptyset$& $\emptyset$& $\emptyset$  &$\emptyset$ &$\alert{t_5^2}, C_7$& $\emptyset$ &$t_6^2, \alert{C_7}$&$\emptyset$&$\cdots$&$\emptyset$&$\emptyset$&$\cdots$&$\emptyset$&$\emptyset$&$\cdots$&$\emptyset$\\\myline
20& $\emptyset$    & $\emptyset$     & $\emptyset$ & $\emptyset$     & $\emptyset$ & $\emptyset$      & $\alert{t_4^4}, C_8$ &$\emptyset$       & $\alert{t_5^4}, C_8$   & $\emptyset$  & $t_6^4, \alert{C_8}$ &$\emptyset$ &$\emptyset$&$\cdots$&$\emptyset$&$\emptyset$&$\cdots$&$\emptyset$&$\emptyset$&$\cdots$&$\emptyset$\\\myline

21& $\emptyset$    & $\emptyset$     & $\emptyset$ & $\emptyset$     & $\emptyset$ & $\emptyset$      & $\emptyset$ &$\emptyset$       &$\emptyset$& $\emptyset$  & $\emptyset$&$\emptyset$ &$\emptyset$&$\cdots$&$\emptyset$&$\emptyset$&$\cdots$&$\emptyset$&$\emptyset$&$\cdots$&$\emptyset$\\\myline
$\vdots$& $\vdots$    & $\vdots$     & $\vdots$ & $\vdots$     & $\vdots$ & $\vdots$      & $\vdots$ &$\vdots$       &$\vdots$& $\vdots$  & $\vdots$&$\vdots$ &$\vdots$&$\vdots$&$\vdots$&$\vdots$&$\vdots$&$\vdots$&$\vdots$&$\vdots$&$\vdots$\\\myline
38& $\emptyset$    & $\emptyset$     & $\emptyset$ & $\emptyset$     & $\emptyset$ & $\emptyset$      & $\emptyset$ &$\emptyset$       &$\emptyset$& $\emptyset$  & $\emptyset$&$\emptyset$ &$\emptyset$&$\cdots$&$\emptyset$&$\emptyset$&$\cdots$&$\emptyset$&$\emptyset$&$\cdots$&$\emptyset$\\\myline
\end{tabular}
}
\end{table*}

Finally, for each clause $C_p$, we set $C_p(2\lambda+ p, c_1) = C_p(2\lambda+ p, c_2) = C_p(2\lambda+ p, c_3) = 1$ where the cells $(2\lambda+p, c_1)$, $(2\lambda+p, c_2)$ and $(2\lambda+p, c_3)$ are already set to value $1$ for the transfer agents of the literals in $C_p$. We also set $C_p(1, 4\lambda + p)= 1$.  These are the \emph{top} cells of the clause agents. All the other cells evaluate to $0$ for the clause agents. 

Now, we show that if $\theta$ is satisfiable, then the Emax value for the reduced LSA instance is $2$. Consider a satisfying assignment of $\theta$. 
First, allocate all the top cells to their respective transfer and clause agents. 
Now, if $x_k$ is assigned true, we allocate the cells $(2k-1, 2k-1)$ and $(2k, 2k)$ to the variable agent $x_k$, and to the transfer agents $t_k^2$ and $t_k^4$, we allocate their bottom cells. To the transfer agents $t_k^1$ and $t_k^3$, we allocate the cells $(2k-1, 2k)$ and $(2k-1, 2k)$.
Else if $x_k$ is assigned $0$, we allocate the cells $(2k-1, 2k)$ and $(2k-1, 2k)$ to the variable agent $x_k$, and to the transfer agents $t_k^1$ and $t_k^3$ we allocate their bottom cells. To the transfer agents $t_k^2$ and $t_k^4$, we allocate the cells $(2k-1, 2k-1)$ and $(2k, 2k)$. 
\cref{tab:fixed-EF1-hard} depicts an example of this allocation.
Since each clause has at least one literal that is assigned to true, at least one bottom cell of a transfer agent must remain unallocated for each clause. We allocate that cell to the clause agent. Thus, each agent gets two cells worth $1$ each.

For the other direction, suppose that there is an allocation for the reduced LSA instance such that the egalitarian social welfare is (at least) $2$. Then, each variable agent $x_k$ must be allocated either ``$(2k-1, 2k)$ and $(2k, 2k-1)$'' or ``$(2k-1, 2k-1)$ and $(2k, 2k)$''. We assign variable $x_k$ to true if and only if agent $x_k$ is allocated $(2k-1, 2k)$ and $(2k, 2k-1)$.
Accordingly, the corresponding transfer agents must be allocated their bottom cells. For a clause agent to have two valued items, it must be allocated a bottom cell of one of the three transfer agents corresponding to its literals. Then, one literal in the clause must have been assigned true. 

It can also be seen that the Umax value is (at least) $2n=38 \lambda /3$ if and only if $\theta$ is satisfiable.  

For the complete settings, the reduction is similar, but we introduce $19\lambda/3$ dummy agents. The matrix also doubles in length and breadth. We set the valuations corresponding to the 3SAT formula as above. Also, for the dummy agents, all the cells of the first two rows give a value of 1. All agents value the rest of the cells at $0$.
Then, by \Cref{lem:extend}, the above allocation can be completed.
\end{proof}

Moreover, the approximation of the Emax LSA problem seems difficult because an $\alpha$-approximation algorithm for the Emax LSA problem implies an $\alpha$-approximation algorithm for the max-min fair allocation problem.
The max-min fair allocation is a problem of maximizing egalitarian welfare when allocating $m$ items to $n$ agents with additive valuations, without any constraints.
The best-known approximation algorithm for the max-min fair allocation problem is $\tilde{O}(m^\epsilon)$-approximation with a running time of $O(m^{1/\epsilon})$, where $m$ is the number of items~\cite{chakrabarty2009allocating}.
\begin{theorem}\label{thm:Emax-reduce}
    There exists an approximation-preserving reduction from the max-min fair allocation problem to the Emax LSA problem.
\end{theorem}
\begin{proof}
As the max-min fair allocation problem, suppose that we are given $n$ agents $[n]$, $m$ items $E=\{e_1,\dots,e_m\}$, and additive utility functions $u_i\colon 2^E\to\mathbb{Z}_+$ for $i\in [n]$.
Without loss of generality, we may assume that $m\ge n$, since otherwise egalitarian social welfare is $0$ for every allocation.
Let $h=\min_{i\in[n]}u_i(E)$. Then, the optimum value for the max-min fair allocation problem is at most $h$.
We construct an LSA problem with $N=M=R=[2m]$ by defining valuations $(v_{ijk})_{i\in N,\,j\in M,\,k\in R}$ as follows:
\begin{align*}
    v_{ijk}&=
    \begin{cases}
        u_i(\{e_j\}) &\text{if }i\in [n]\text{ and }j=k\in[m],\\
        h            &\text{if }i\ge n+1,\\
        0            &\text{otherwise}\\
    \end{cases}
    &\hspace{-15mm} (i\in N,\ j\in M,\ k\in R).
\end{align*}
We prove that the optimum value of the max-min fair allocation problem is the same as the Emax value of the reduced LSA problem.

Let $(X_1,\dots,X_n)$ be a partition of $E$.
Then, consider a complete allocation $A\subseteq N\times M\times R$ such that $A(j,j)=i$ if $e_j\in X_i$.
Note that such a complete allocation must exist by \Cref{lem:extend}.
The egalitarian social welfare of $A$ in the LSA problem is $\min_{i\in [n]}u_i(X_i)$ because 
$v_i(A_i)=u_i(X_i)$ for each $i\in[n]$ and 
$v_i(A_i)\ge m\cdot h\ge h$ for each $i\in[2m]\setminus [n]$.

Conversely, let $A\subseteq N\times M\times R$ be a (possibly partial) allocation.
We construct a complete allocation $\overline{A}\subseteq N\times M\times R$ such that $\overline{A}(j,j)=i$ if $A(j,j)=i$ for $i\in[n]$ and $j\in[m]$. Such a complete allocation must exist by \Cref{lem:extend}.
Then, the egalitarian social welfare of $A$ and $\overline{A}$ are $\min_{i\in[n]}v_i(A_i)~(\le h)$ because
$v_{i}(\overline{A}_i)=v_i(A_i)$ for each $i\in[n]$ and
$v_i(\overline{A}_i)\ge m\cdot h\ge h$ for each $i\in[2m]\setminus [n]$.
Let $(X_1,\dots,X_n)$ be a partition of $E$ such that $e_j\in X_i$ if $\overline{A}(j,j)=i$.
Then, $u_i(X_i)\ge v_i(\overline{A}_i)$ for every $i\in[n]$.
\end{proof}

Finally, we prove that it is strongly NP-complete to check the existence of an EF, EQ, PROP, EFX, EQX, or PROPX complete allocation by a reduction from the $3$-Partition problem.
This hardness for the LSA problem holds even when the valuations are identical.
It is worth mentioning that, when the valuations are identical, any complete allocation maximizes utilitarian social welfare.
Moreover, if all items or all rounds are identical, any complete allocation satisfies Umax, Emax, and all other fairness properties.
\begin{theorem}\label{thm:EF}
    Even when the valuations are identical, checking the existence of a complete allocation in an LSA problem that satisfies each of EF, PROP, EQ, EFX, EQX, and PROPX is strongly NP-complete.
Moreover, even when the valuations are identical, checking whether the Emax value of a complete LSA problem is at least a certain value is also strongly NP-complete.
\end{theorem}
\begin{proof}
   We present a polynomial-time reduction from the $3$-Partition problem, which is known to be NP-complete~\cite{garey1979computers}. In the 3-Partition Problem, we are given $3m$ positive integers $a_1,a_2,\ldots,a_{3m}$ such that $T/4<a_j<T/2$ and $\sum_{j=1}^{3m} a_j=mT$. The goal is to determine whether or not there is a partition $(S_1,\dots,S_m)$ of the set $[3m]$ such that $|S_i|=3$ and $\sum_{j\in S_i}a_j=T$ for each $i\in[m]$.  

From a given instance of the 3-Partition problem, we construct an LSA problem with $N=M=R=[6m]$ and valuations $(v_{ijk})_{i\in N,\,j\in M,\,k\in R}$ as follows:
\begin{align*}
v_{ijk}&=
\begin{cases}
     a_j & \text{if } j=k\le 3m,\\
     T   & \text{if } j=3m-1\text{ and }k\leq 2m+1,\\
     T   & \text{if } j=3m  \text{ and }k\leq 3m-1,\\
     0   & \text{otherwise}
\end{cases}&(i\in N,\ j\in M,\ k\in R).
\end{align*}
Note that the valuations are identical among the agents $N$. 
\begin{table}[htbp]
\caption{The valuation of item $j$ for $k$th round defined in the proof of \Cref{thm:EF}}\label{latinEXPROPEQ}
\centering
\tabcolsep = 1mm
\renewcommand{\noline}[1]{\multicolumn{1}{p{6.8mm}}{\centering\textcolor{blue}{\ssmall #1}}}
\newcommand{\myline}{\hhline{~-------------}}
\begin{tabularx}{\linewidth}{T|c|c|c|c|c|c|c|c|c|c|c|c|c|}
\noline{}& \noline{$1$} & \noline{$2$} & \noline{$3$} & \noline{$\cdots$} & \noline{$2m$} & \noline{$2m{+}1$} & \noline{$2m{+}2$}& \noline{$\cdots$}  & \noline{$3m{-}1$} & \noline{$3m$} & \noline{$3m{+}1$} & \noline{$\cdots$} & \noline{$6m$}\\\myline
$1$      & $a_1$ & $0$ & $0$ & $\cdots$ & $0$ & $0$ & $0$& $\cdots$ & $0$ & $0$ & $0$ & $\cdots$ & $0$\\\myline
$2$      & $0$ & $a_2$ & $0$ & $\cdots$ & $0$ & $0$ & $0$& $\cdots$ & $0$ & $0$ & $0$ & $\cdots$ & $0$\\\myline
$3$      & $0$ & $0$ & $a_3$ & $\cdots$ & $0$ & $0$ & $0$& $\cdots$ & $0$ & $0$ & $0$ & $\cdots$ & $0$\\\myline
$\vdots$ & $\vdots$ & $\vdots$  & $\vdots$ & $\ddots$ & $\vdots$ & $\vdots$ & $\vdots$& & $\vdots$ & $\vdots$ & $\vdots$ & & $\vdots$\\\myline
$2m$   & $0$ & $0$ & $0$ & $\cdots$ & $a_{2m}$ & $0$ & $0$&       & $0$ & $0$ & $0$ & $\cdots$ & $0$\\\myline
$2m+1$ & $0$ & $0$ & $0$ & $\cdots$ & $0$ & $a_{2m+1}$ & $0$&     & $0$ & $0$ & $0$ & $\cdots$ & $0$\\\myline
$2m+2$ & $0$ & $0$ & $0$ & $\cdots$ & $0$ & $0$ & $a_{2m+2}$&     & $0$ & $0$ & $0$ & $\cdots$ & $0$\\\myline
$\vdots$ & $\vdots$ & $\vdots$  & $\vdots$ &  & $\vdots$ & $\vdots$ & $\vdots$&$\ddots$ & $\vdots$ & $\vdots$ & $\vdots$ & & $\vdots$\\\myline
$3m-1$ & $T$ & $T$ & $T$  & $\cdots$ & $T$ & $T$ & $0$&$\cdots$  & $a_{3m-1}$ & $0$  & $0$ & $\cdots$ & $0$\\\myline
$3m$   & $T$ & $T$ & $T$  & $\cdots$ & $T$ & $T$ & $T$&$\cdots$  & $T$ & $a_{3m}$    & $0$ & $\cdots$ & $0$\\\myline
$3m+1$ & $0$ & $0$ & $0$  & $\cdots$ & $0$ & $0$ & $0$&$\cdots$  & $0$ & $0$ & $0$ & $\cdots$ & $0$\\\myline
$\vdots$ & $\vdots$ & $\vdots$  & $\vdots$ &  & $\vdots$ & $\vdots$ & $\vdots$& & $\vdots$ & $\vdots$ & $\vdots$ & & $\vdots$\\\myline
$6m$   & $0$ & $0$ & $0$  & $\cdots$ & $0$ & $0$ & $0$&$\cdots$  & $0$ & $0$ & $0$ & $\cdots$ & $0$\\\myline
\end{tabularx}
\end{table}

Suppose that the given instance of the 3-Partition problem is a Yes-instance, i.e., there is a 3-partition $(S_1,S_2,\dots,S_m)$ such that $|S_i|=3$ and $\sum_{a_j\in S_i}a_j=T$ for $i\in [m]$. Then, consider a complete allocation $A\subseteq N\times M\times R$ such that $A(j,j)=i$ if $a_j\in I_i$ for each $i\in[m]$ and $j\in [m]$, $A(3m-1,k)=m+k$ for each $k\in[2m+1]$, and $A(3m,k)=3m+1+k$ for each $k\in[3m-1]$. 
Note that such a complete allocation must exist by Lemma~\ref{lem:extend}.
Then, the complete allocation $A$ is EF, EQ, PROP, EFX, EQX, PROPX, and of at least egalitarian social welfare $T$.

Conversely, let $A\subseteq N\times M\times R$ be a complete allocation that is EF, PROP, EQ, EFX, EQX, PROPX, or has egalitarian social welfare of at least $T$.
Then, $v_i(A_i)=T$ for each agent $i\in N$. 
Let $S_i=\{j\in[3m]\mid A(j,j)=i\}$ for each $i\in [6m]$.
By the conditions $T/4<a_j<T/2$ and $\sum_{j=1}^{3m} a_j=mT$, the set $S_i~(i\in[6m])$ has exactly $3$ elements if it is not empty.
Thus, $\mathcal{S}=\{S_i\mid i\in[6m],\,S_i\ne\emptyset\}$ is a 3-partition of $[3m]$ such that $|S_i|=3$ and $\sum_{j\in S_i}a_j=T$ for each $S_i\in\mathcal{S}$.

Therefore, it is NP-hard to decide whether there exists a complete allocation that is EF, EQ, PROP, EFX, EQX, PROPX, or has egalitarian social welfare of at least $T$.
\end{proof}

It is worth mentioning that we can also prove NP-completeness by checking the existence of a complete allocation satisfying the following weaker variants of EFX, EQX, and PROPX: $v_i(A_i)\ge v_i(A_{i'}\setminus\{(j,k)\})$ for any $(j,k)\in A_{i'}$ with $v_{i'jk}>0$,
$v_i(A_i)\ge v_{i'}(A_{i'})-\min_{(j,k)\in A_{i'}:\,v_{i'jk}>0}v_{i'jk}$, and
$v_i(A_i)\ge v_{i}(M\times R)/n-\min_{(j,k)\in A_{i}:\,v_{ijk}>0}v_{ijk}$.
This can be obtained by modifying the value of item $6m$ to a sufficiently small positive value $\epsilon$ in the proof of \Cref{thm:EF}.



\section{Concluding Remarks and Future Work}

In this paper, we introduced problems of allocating indivisible items under the Latin square constraint. 
This approach is effective for numerous practical problems. As demonstrated in \Cref{subsec:application}, it is particularly useful for scheduling sightseeing activities, school timetabling, and job rotation in organizations.
The method ensures fairness, efficiency, and balance, thereby preventing the overuse or redundancy of any single element within a system.

This study investigated the computational complexity of finding a fair or efficient allocation under the Latin square constraint.
We provided $(1-1/e)$- and $(1-1/e)/4$-approximation algorithms for the partial and complete Umax LSA problems, respectively.
Additionally, we presented FPT algorithms with respect to the order of Latin square and the optimum value for both the partial and complete Umax LSA problems.
Regarding impossibility results, we demonstrated the NP-hardness of the Umax and Emax problems.
Furthermore, we proved NP-hardness for various settings, including checking the existence of an EF, PROP, EQ, EFX, EQX, or PROPX complete allocation.

A straightforward direction for future work is to construct algorithms for Umax with improved approximation ratios. Developing faster FPT algorithms for Umax and Emax also presents an interesting avenue for exploration. Another open question is determining the complexity of checking the existence of an EF1, EQ1, or PROP1 complete allocation. 
When valuations are binary, we can easily check whether the Emax value of a partial LSA problem is at least $1$ or at most $0$ by solving an agent-side perfect matching problem on a bipartite graph $\big(N,M\times R;\{(i,(j,k))\mid v_{ijk}=1\}\big)$. However, this problem remains unresolved for complete allocations since a partial allocation corresponding to a perfect matching may not be completed (see the rightmost partial Latin square example in Introduction).

\section*{Acknowledgment}
This work was partially supported by 
JST ERATO Grant Number JPMJER2301, 
JST PRESTO Grant Number JPMJPR2122, 
JSPS KAKENHI Grant Number JP20K19739,
Value Exchange Engineering, a joint research project between Mercari, Inc.\ and the RIISE, Sakura Science Exchange Program, and SERB MATRICS Grant Number MTR/2021/000474.





\bibliographystyle{plainnat}
\bibliography{sn-bibliography}

\end{document}